\documentclass[%
 reprint,
 amsmath,amssymb,
 aps,
 pra,
 nofootinbib
]{revtex4-2}

\usepackage[
    unicode=true,
    colorlinks=true,
    linkcolor=blue,
    citecolor=blue,
    urlcolor=blue
]{hyperref}

\usepackage[english]{babel}
\usepackage{xcolor}
\usepackage{amsthm}
\usepackage{bm}
\usepackage{mathtools}
\newtheorem{definition}{Definition}
\newtheorem{theorem}{Theorem}
\newtheorem{proposition}{Proposition}

\begin{document}

\title{New bound on \texorpdfstring{$S_{1}\times S_{2}$}{S1 x S2}-setting
Bell locality of a nonseparable Werner state}
\author{Elena R. Loubenets$^{1,2}$\\$^{1}$Department of Applied Mathematics, MIEM, HSE University, \\Moscow 123458, Russia\\$^{2}$Steklov Mathematical Institute of Russian Academy of Sciences, \\Moscow 119991, Russia}

\begin{abstract}
In many quantum applications it is important to know whether or not a Bell nonlocal two-qudit state exhibits its nonlocality under correlation scenarios
with some given numbers $S_{1},S_{2}\geq1$ of generalized quantum measurements at
two sites. In the present article, we find analytically a new general condition sufficient for a 
nonseparable Werner state with  a dimension $d\leq\min\{S_{1},S_{2}\}$ to satisfy all Bell inequalities under every $S_{1}\times S_{2}$-setting correlation
scenario with outcomes of an arbitrary spectral type, discrete or continuous --  that is, to be $S_{1}\times S_{2}$-setting Bell local, for short. For a variety of measurement settings $S_{1}, S_{2}$, this new general bound is beyond (i) Werner’s and Barrett’s locality bounds for a nonseparable Werner state  and (ii) the related result by Johnson and Viola [\emph{Phys. Rev. A} \textbf{88}, 032323 (2013)] on "sharability" of this state. When specialized to $d = 2$ and $\min\{S_{1},S_{2}\} = 2, 3$, our new general bound improves upon the corresponding results by Terhal et. al. in [\emph{Phys. Rev. Lett.}
\textbf{90,} 157903 (2003)], moreover, for $d=\min\{S_{1},S_{2}\} = 2$,  it reduces exactly to the bound derived for this case via the CHSH inequality. We prove also explicitly in operator terms the result of Terhal et. el. that every nonseparable Werner state with a dimension $d >\min\{S_{1},S_{2}\}$ is $S_{1}\times S_{2}$-setting Bell local. The new results presented in this article are important for both the theory of Bell nonlocality and quantum applications based on this fundamental concept. 
\end{abstract}

\maketitle
\section{Introduction}
Since the seminal papers of Bell \cite{1,2}, nonlocality of a bipartite
quantum state -- in the sense of its violation of a general Bell inequality,
\emph{Bell nonlocality,} has been studied by many authors, see \cite{3,4,5,6}
and references therein. If a bipartite quantum state satisfies all general
Bell inequalities, then it is referred to as Bell
local. On the relation between Bell locality and
Einstein--Podolsky--Rosen (EPR) \cite{7} locality, see \cite{3,6}.

Analysing violation the Clauser--Horne--Shimony--Holt (CHSH) inequality
\cite{2'} by a pure two-qudit state, Gisin and Peres \cite{8} proved that
every nonseparable pure two-qudit state is Bell nonlocal. Recently, this fact
was re-established in \cite{9} via the construction in terms of the pure state
concurrence of the lower and upper bounds on the maximal violation by a pure
two-qudit state of all general Bell inequalities.

However, nonseparability of a mixed two-qudit state does not imply its 
nonlocality and, in 1989, Werner introduced \cite{10} an example of mixed
two-qudit states, \emph{the Werner states,} and proved that, under a bipartite
correlation scenario with an arbitrary number of projective measurements at
each of sites, some of nonseparable Werner states admit the local hidden
variable (LHV) model constructed in \cite{10}, thus, exhibit Bell locality
beyond separability.

In 2002, Barrett \cite{11} constructed another LHV model for the Werner states and presented the condition sufficient for a nonseparable Werner state to admit this LHV model under correlation scenarios with an arbitrary number of generalized measurements at each of sites. 

For the review of the LHV models for other nonseparable
mixed two-qudit states, in particular, noisy states, see \cite{12,13,14} and
references therein.

Note that Werner's \cite{10} and Barrett's \cite{11} locality bounds are only sufficient but not necessary for locality
of nonseparable Werner states, so that a nonseparable Werner
state may be Bell local beyond these locality bounds.

More generally, Bell nonlocality of a nonseparable two-qudit state does not
exclude the situation where, under any correlation scenario with some specific
numbers, let $S_{1}$ and $S_{2},$ of generalized quantum measurements at each of
sites, this nonlocal two-qudit state behaves itself as Bell local -- in the sense that it 
satisfies all general Bell inequalities specified for these numbers
of measurement settings. In \cite{3,4,6}, we refer to this  type of Bell locality as the
$S_{1}\times S_{2}$-\emph{setting Bell locality } 

The study of $S_{1}\times
S_{2}$\emph{-}setting Bell locality was considered first
in\footnote{Independently and via different notions. For the comparison, see
Section II.C.} \cite{15,16} and further developed in the series of author's
articles on Bell nonlocality, see \cite{4,6} and references therein.

Nonlocality of Werner states is now used in many quantum applications and it
is important to know  whether or not a
nonlocal Werner state really exhibits its nonlocality  under every $S_{1}\times S_{2}$-setting
correlation scenario with generalized measurements at each of sites.  

Besides Werner's and Barrett's locality bounds for nonseparable Werner
states, it is also known due to the results in \cite{15} that
every nonseparable Werner state of a dimension $d>\min\{S_{1},S_{2}\}$ is
$S_{1}\times S_{2}$-setting Bell local. It was also proved
\cite{17} that every nonseparable Werner state of dimension $d\geq3$
satisfies \cite{17} the perfect correlation form of the original Bell
inequality without necessarily obeying the condition on perfect correlations
\cite{18}.

In the present article, we find analytically a new general condition
sufficient for $S_{1}\times
S_{2}$-setting Bell locality of a nonseparable Werner state with a dimension $d\leq\min
\{S_{1},S_{2}\}$. For a variety of measurements settings 
$S_{1},S_{2}\geq1$, this new general bound is beyond (i) Werner's \cite{10} and Barrett's \cite{11} 
locality bounds for a nonseparable Werner state and (ii) the related result by Johnson and Viola \cite{23} on "sharability"\footnote{Here, we retain the terminology of \cite{23}.} of Werner states. When specialized to $d = 2$ and $\min\{S_{1},S_{2}\} = 2, 3$, our new general bound improves upon the corresponding results by  Terhal et. al. in [16], moreover, for $d=\min\{S_{1},S_{2}\} = 2$,  it reduces exactly to the bound derived for this case via the CHSH inequality. We prove also explicitly in
operator terms the result of Terhal et. el. \cite{15} that every nonseparable
Werner state with a dimension $d>\min\{S_{1},S_{2}\}$ is $S_{1}\times S_{2}%
$-setting Bell local.

\section{Preliminaries: source operators, tensor positivity, sufficient conditions}

In this Section, for our further consideration in Section III, we shortly
specify the main notions and statements, introduced originally in \cite{16,
17} and generalized/developed further in \cite{4} for finding a general analytical upper bound\footnote{See Eq.(43) in \cite{4}.} on the maximal violation by
an $N$-partite quantum state of all general Bell inequalities. The latter
bound implies in turn a general condition (Proposition 5
of [4]) sufficient for $S_{1}\times\cdots\times
S_{N}$\emph{-setting Bell locality of an N-partite quantum state}.  Specified for a bipartite case, this general condition extends the
settings of Theorems 1, 2 in \cite{15}, for the comparison see below Section II.C and Section II.D.

\subsection{Source operators}

For a quantum state $\rho$ on a complex finite-dimensional Hilbert space
$\mathcal{H}_{d_{1}}\otimes\mathcal{H}_{d_{2}}$, $d_{i}=\dim\mathcal{H}%
_{d_{i}}\geq2,$ $i=1,2,$ and positive integers $S_{1},S_{2}\geq1$, let
$T_{S_{1}\times S_{2}}^{(\rho)}$ be a Hermitian operator on $\mathcal{H}%
_{d}^{\otimes S_{1}}\otimes\mathcal{H}_{d}^{\otimes S_{2}}$, satisfying the
relation%
\begin{align}
&  \mathrm{tr[}T_{S_{1}\times S_{2}}^{(\rho)}\{\mathbb{I}_{\mathcal{H}_{d_{1}%
}^{\otimes k_{1}}}\otimes X_{1}\otimes\mathbb{I}_{\mathcal{H}_{d_{1}}%
^{\otimes(S_{1}-1-k_{1})}}\label{1}\\
&  \otimes\mathbb{I}_{\mathcal{H}_{d_{2}}^{\otimes k_{2}}}\otimes X_{2}%
\otimes\mathbb{I}_{\mathcal{H}_{d_{2}}^{\otimes(S_{2}-1-k_{2})}}\}]\nonumber\\
&  =\mathrm{tr}\left[  \rho\left\{  X_{1}\otimes X_{2}\right\}  \right]
,\nonumber\\
k_{1}  &  =0,...,(S_{1}-1),\text{ \ \ }k_{2}=0,...,(S_{2}-1),\nonumber
\end{align}
for any linear operators $X_{1},X_{2}$ on $\mathcal{H}_{d_{1}}$ and
$\mathcal{H}_{d_{2}},$ respectively. Here, $\mathbb{I}_{\mathcal{H}%
_{d}^{\otimes k}}\otimes X\mid_{_{k=0}}:=$ $X\otimes\mathbb{I}%
_{\mathcal{H}_{d}^{\otimes k}}\mid_{_{k=0}}:=X.$ Clearly,
$\mathrm{tr}[T_{S_{1}\times S_{2}}^{(\rho)}]=1$\ and $T_{_{1\times1}}^{(\rho
)}\equiv\rho.$ 

Relation (\ref{1})\ implies that the partial trace of operator
$T_{S_{1}\times S_{2}}^{(\rho)}$ on $\mathcal{H}_{d}^{\otimes S_{1}}%
\otimes\mathcal{H}_{d}^{\otimes S_{2}}$ over any selection of $(S_{1}-1)$
Hilbert spaces $\mathcal{H}_{d_{1}}$ in block $\mathcal{H}_{d_{1}}^{\otimes
S_{1}}$ and any selection of $(S_{2}-1)$ Hilbert spaces $\mathcal{H}_{d_{2}}$
in block $\mathcal{H}_{d}^{\otimes S_{2}}$ is invariant of a choice of these
selections and is equal to:%
\begin{align}
\MoveEqLeft[1.5]\mathrm{tr}_{\mathcal{H}_{d_{1}}^{\otimes(S_{1}-1)},\mathcal{H}_{d_{2}%
}^{\otimes(S_{2}-1)}}[T_{S_{1}\times S_{2}}^{(\rho)}]    =\rho,\label{2} \\
\MoveEqLeft[1.5]\forall\text{ }\mathcal{H}_{d_{i}}^{\otimes(S_{i}-1)}    \subset
\mathcal{H}_{d_{i}}^{\otimes S_{i}},\text{ }i=1,2.\nonumber
\end{align}

\begin{definition}
\label{source-operator} A Hermitian operator $T_{S_{1}\times S_{2}}^{(\rho)}$
on $\mathcal{H}_{d_{1}}^{\otimes S_{1}}\otimes\mathcal{H}_{d_{2}}^{\otimes
S_{2}}$ satisfying relation (\ref{1}) is called \cite{16,17,4} a $S_{1}\times
S_{2}$-setting source operator for a state $\rho$ on $\mathcal{H}_{d_{1}%
}\otimes\mathcal{H}_{d_{2}}.$ If a source operator $T_{S_{1}\times S_{2}%
}^{(\rho)}$ for a state $\rho$ is positive, then it constitutes an extension
of state $\rho$ to space $\mathcal{H}_{d_{1}}^{\otimes S_{1}}\otimes\mathcal{H}%
_{d_{2}}^{\otimes S_{2}}.$
\end{definition}

By Propositions 1 in \cite{16,4} for every state $\rho$ on $\mathcal{H}%
_{d_{1}}\otimes\mathcal{H}_{d_{2}}$ and arbitrary positive integers
$S_{1},S_{2}\geq1$ a source operator $T_{S_{1}\times S_{2}}^{(\rho)}$ exists. 

The convex sum of $S_{1}\times
S_{2}$-setting source operators for a state $\rho$ is also the $S_{1}\times
S_{2}$-setting source operator for this state. 

If a Hermitian operator on a Hilbert space $\left(  \mathcal{H}_{d}\right)
^{\otimes L_{1}}\otimes\left(  \mathcal{H}_{d}\right)  ^{\otimes L_{2}}$,
$1\leq L_{n}'\leq S_{n},$ $n=1,2,$ is reduced from a source operator $T_{S_{1}%
\times S_{2}}^{(\rho)}$ on $\left(  \mathcal{H}_{d}\right)  ^{\otimes S_{1}%
}\otimes\left(  \mathcal{H}_{d}\right)  ^{\otimes S_{2}}$, then it constitutes
\cite{4} a $L_{1}\times L_{2}$-setting source operator for state $\rho.$

The notion of a sharing state, introduced in \cite{23}, corresponds to a particular class of source operators, namely, positive source operators.

We also stress that by Definition
\ref{source-operator} a source operator $T_{_{S_{1}\times S_{2}}%
}^{(\rho)}$ does not need to be invariant under permutations of tensor factors
in $\mathcal{H}_{d_{1}}^{\otimes S_{1}}$ and in $\mathcal{H}_{d}^{\otimes
S_{2}}.$ Take, for example, the Hermitian operator $\frac{1}{d_{1}%
d_{2}r_{d_{2}}^{(+)}}\mathbb{I}_{d_{1}}\otimes\mathrm{P}_{d_{2}}^{(+)}$
$\otimes\mathbb{I}_{d_{2}}$ on $\mathcal{H}_{d_{1}}\otimes\mathcal{H}_{d_{2}}^{\otimes3},$ where $\mathrm{P}_{d_{2}}^{(+)}$ is
the projection with rank $r_{d_{2}}^{(+)}=\frac{d_{2}(d_{2}+1)}{2}$ onto the
symmetric subspace of $\mathbb{I}_{d_{2}}^{\otimes2}.$ This Hermitian operator
is not invariant under permutations of spaces $\mathcal{H}_{d_{2}}$ in
$\mathcal{H}_{d_{2}}^{\otimes3}$. However, the partial trace of this operator
over a selection of two spaces $\mathcal{H}_{d_{2}}$ in $\mathcal{H}_{d_{2}%
}^{\otimes3}$ is the same for any of possible choices and is equal to
$\frac{\mathbb{I}_{d_{1}}\otimes\mathbb{I}_{d_{2}}}{d_{1}d_{2}}.$ Therefore,
by Definition 1, this Hermitian operator constitutes a $1\times3$--setting positive
source operator $T_{1\times3}^{(\rho)}$ for the maximally mixed state
$\rho=\frac{\mathbb{I}_{d_{1}}\otimes\mathbb{I}_{d_{2}}}{d_{1}d_{2}}.$

\subsection{Tensor positivity}

The following general notion is introduced in Section 2 of \cite{4}. Here, we
consider only a finite-dimensional case.

\begin{definition}
\label{tensor-positive} Let $W$ be a Hermitian operator on a Hilbert space
$\mathcal{H}_{d_{1}}\otimes\cdots\otimes\mathcal{H}_{d_{m}},$ $m\geq2.$ If, for operator $W$,  the scalar product
\begin{equation}
\langle \psi _{1}\otimes \cdots \otimes \psi _{m}|\;W|\psi_{1}\otimes \cdots \otimes \psi _{m}\rangle \geq 0,\label{3}%
\end{equation}
for all $\psi_{j}\in\mathcal{H}_{d_{j}},\; j=1,...,m,$ then it
is called \cite{4}  tensor positive and this is denoted by $W\overset{\otimes}{\geq}0.$
\end{definition}

Every positive operator on $\mathcal{H}_{d_{1}}\otimes\cdots\otimes\mathcal{H}_{d_{m}},$ $m\geq2,$ is tensor positive, but not vice versa.

For each tensor positive operator $W\overset{\otimes}{\geq}0,$ relation
$\mathrm{tr}\left[  W\{X_{1}\otimes\cdots\otimes X_{m}\}\right]  \geq0$ holds
for any positive operators $X_{1},...,X_{m}$ on spaces $\mathcal{H}_{d_{1}},...,\mathcal{H}_{d_{m}}$, respectively. In particular, $\mathrm{tr}[W]\geq0. $

We stress that, for an arbitrary two-qudit state $\rho$, a $S_{1}\times S_{2}%
$-setting source operator $T_{S_{1}\times S_{2}}^{(\rho)}$ exists\emph{\ }%
\cite{16, 4}, however, \emph{does not need to be tensor positive. }By Lemma 1
in \cite{4} we have the following general statement.
\begin{proposition}
\label{reduced} If operator $T_{S_{1}\times S_{2}}^{(\rho)}$ on $\mathcal{H}%
_{d_{1}}^{\otimes S_{1}}\otimes\mathcal{H}_{d_{2}}^{\otimes S_{2}}$ is a
tensor positive source operator for a state $\rho$ on 
$\mathcal{H}_{d_{1}}\otimes\mathcal{H}_{d_{2}},$ then the source operator
$T_{L_{1}\times L_{2}}^{(\rho)}$ on a Hilbert space $\mathcal{H}_{d}^{\otimes
L_{1}}\otimes\mathcal{H}_{d}^{\otimes L_{2}}$, $1\leq L_{i}\leq S_{n},$
$L_{1}+L_{2}<S_{1}+S_{2},$ reduced from $T_{S_{1}\times S_{2}}^{(\rho)}$ is
also tensor positive, the converse is not true.
\end{proposition}

\subsection{Tensor positive source operators versus  symmetric quasi-extensions and "sharable" states}

As mentioned above, for an arbitrary state $\rho$ on $\mathcal{H}_{d_{1}}\otimes\mathcal{H}_{d_{2}},$ a
source operator $T_{S_{1}\times S_{2}}^{(\rho)}$ on space $\mathcal{H}%
_{d}^{\otimes S_{1}}\otimes\mathcal{H}_{d}^{\otimes S_{2}}$ exists but does
not need to be: (i) invariant under permutations of spaces in blocks
$\mathcal{H}_{d_{1}}^{\otimes S_{1}}$ and $\mathcal{H}_{d_{2}}^{\otimes S_{2}}%
$(see example in Section II.A) and (ii) tensor positive. Furthermore, a tensor
positive source operator may be positive, therefore, it does not need to be an
entanglement witness.

All this implies that the notion of a $S_{1}\times S_{2}$-setting source
operator and the notion of a tensor positive $S_{1}\times S_{2}$-setting source
operator introduced in \cite{16} and further developed in \cite{4} are more general than the notions
of a symmetric $(S_{1}, S_{2})$ extension and a symmetric
$(S_{1},S_{2})$ quasi-extension of a state $\rho$ introduced in
\cite{15}. The latter constitute particular types of tensor positive
$S_{1}\times S_{2}$-setting source operators for a state $\rho$.

The notion of a source operator \cite{16,4} is also more general than the notion of a "sharable" state in \cite{23} since a "sharable" state constitutes a source operator which is positive. 

\subsection{Sufficient conditions}

Let $\rho$ be a two-qudit state on a finite-dimensional complex Hilbert space
$\mathcal{H}_{d_{1}}\otimes\mathcal{H}_{d_{2}},$ $d_{i}\geq2,$ $i=1,2,$ and
$S_{1},S_{2}\geq1$ be arbitrary positive integers. The following general
concepts and sufficient conditions are introduced in \cite{3,4}.

\begin{definition}
\label{lhv} A state $\rho$ admits the $S_{1}\times S_{2}$-setting LHV
description if every $S_{1}\times S_{2}$-setting correlation scenario with an
arbitrary type of parties' quantum measurements on state
$\rho$ admits an LHV model not necessarily the same for all $S_{1}\times
S_{2}$-setting scenarios.
\end{definition}

According to Proposition 3 in \cite{3}, if a two-qudit state $\rho$ admits a
$S_{1}\times S_{2}$-setting LHV description, then it also admits $L_{1}\times
L_{2}$-setting LHV description for all $L_{i}\leq S_{i},i=1,2.$

\begin{definition}
\label{Bell-locality} A state $\rho$ is a $S_{1}\times S_{2}$-setting Bell
local if it does not violate \cite{6} any of general $L_{1}\times L_{2}%
$\emph{-setting Bell inequalities} with\emph{\ }$L_{i}\leq S_{i},$
$i=1,2.$\emph{\ }
\end{definition}

By Proposition 6 in \cite{4} the following statements are mutually
equivalent: (i) state $\rho$ admits the $S_{1}\times S_{2}$-setting LHV
description; (ii) state $\rho$ does not violate any of $L_{1}\times L_{2}%
$-setting Bell inequalities where $L_{i}\leq S_{i},i=1,2.$  This and Definitions \ref{lhv}, \ref{Bell-locality} imply.

\begin{proposition}
\label{Lhv-Bell} A state $\rho$ is $S_{1}\times S_{2}$-setting Bell local iff
it admits the $S_{1}\times S_{2}$-setting LHV description.\medskip
\end{proposition}

By specifying for a bipartite case the setting of Proposition 5 in \cite{4} on
$S_{1}\times\cdots\times S_{N}$-Bell locality of an $N$-partite state and
taking into account the above Proposition \ref{Lhv-Bell} we come to the following general
Theorem.\medskip

\begin{theorem}
\label{theorem-1} Let $\rho$ be a quantum state on a complex finite-dimensional
Hilbert space $\mathcal{H}_{d_{1}}\otimes\mathcal{H}_{d_{2}},$ $d_{i}\geq2,$
$i=1,$ and\ $S_{1},S_{2}$ be arbitrary positive integers. \newline(1) If, for state $\rho$, there exists a tensor positive source operator $T_{1\times S_{2}%
}^{(\rho)}$, then this state is $S_{1}\times S_{2}$-setting Bell local\footnote{According to Definition \ref{Bell-locality}, if state $\rho$ is $S_{1}\times S_{2}$-setting Bell local, then it is also $L_{1}\times L_{2}$-setting Bell local for any $L_{i}\leq S_{i}$.}, 
moreover, is $L_{1}^{\prime}\times S_{2}$-setting Bell local for an arbitrary
number $L_{1}^{\prime}>S_{1}$ of generalized measurements at site $"1"$.
\smallskip\newline(2) If, for  state $\rho$, there exists a tensor positive source
operator $T_{S_{1}\times1}^{(\rho)}$, then this state is $S_{1}\times S_{2}%
$-setting Bell local, moreover, is $S_{1}\times L_{2}^{\prime}$-setting Bell
local for an arbitrary number $L_{2}^{\prime}>S_{2}$ \ of generalized
measurements at site $"2"$.
\smallskip\newline(3) If, for state $\rho$, there exists a
tensor positive source operator $T_{S_{1}\times S_{2}}^{(\rho)},$ then this
state is $S_{1}\times S_{2}$-setting Bell local, moreover, for this state,
there exist tensor positive source operators $T_{1\times S_{2}}^{(\rho)}$ and
$T_{S_{1}\times1}^{(\rho)}$ and the statements of items (1) and (2) hold.
\end{theorem}

As we point out above in Section II.C, the term "a tensor positive source
operator" includes the notion of a symmetric quasi-extension \cite{15} and a "sharable" state only as
particular cases. Due to this and Definitions 3,4, the setting of the above
Theorem 1 is more general than the settings of Theorems 1, 2 in \cite{15}
formulated in terms of symmetric quasi-extensions.

\section{Bell locality of nonseparable Werner states}

In this Section, we present bounds on $S_{1}\times S_{2}$-setting Bell
locality of Werner states \cite{10} on $\mathcal{H}_{d}\otimes
\mathcal{H}_{d}$. These states have the form 
\begin{equation}
W_{d,\Phi}=\frac{1+\Phi}{2}\frac{\mathrm{P}_{d,2}^{(+)}}{r_{d}^{(+)}}%
+\frac{1-\Phi}{2}\frac{\mathrm{P}_{d,2}^{(-)}}{r_{d}^{(-)}},\text{ \ \ }%
\Phi\in\lbrack-1,1], \label{4}%
\end{equation}
where $\mathrm{P}_{d,2}^{(\pm)}$ are the orthogonal projections onto the
symmetric (plus sign) and antisymmetric (minus sign) subspaces of
$\mathcal{H}_{d}^{\otimes2}$ of dimensions $r_{d}^{(\pm)}=\frac{d(d\pm1)}{2}.$
A Werner state $W_{d,\Phi}$ is separable iff parameter $\Phi
\in\lbrack0,1]$ and nonseparable otherwise. 

As proved by Werner \cite{10},
under any bipartite correlation scenario with an arbitrary number of projective
measurements at each of sites, a nonseparable Werner state $W_{d,\Phi}$
admits the LHV model constructed in \cite{10} (thus, Bell local) if
\begin{equation}
\Phi\in\lbrack-\gamma_{pm}^{(\text{W})},0\text{ }),\ \ \ \ \ \gamma_{pm}%
^{(\text{W})}(d):=1-\frac{d+1}{d^{2}}. \label{5}%
\end{equation}

In 2002, Barrett \cite{11} presented another LHV model for Werner states and
proved that, under a correlation scenario with an arbitrary number of generalized measurements at each of sites, a
nonseparable Werner state $W_{d,\Phi}$ admits this model (thus, Bell local)
if
\begin{equation}
\Phi\in\lbrack-\gamma_{gm}^{(\text{B})},0)\text{, \ \ }\gamma_{gm}^{(\text{B})}%
(d):=\frac{3d-1}{d^{2}}\left(\frac{d-1}{d}\right)^{d-1}-\frac{1}{d}.
\label{6}%
\end{equation}
Note that the locality bounds in (\ref{5}) and in (\ref{6}), satisfying relation
$\gamma_{pm}^{(\text{W})}(d)$ $>\gamma_{gm}^{(\text{B})}(d)$ for all $d\geq2$, behave
themselves quite differently with the growth of a dimension $d$ $-$ the Werner
bound $\gamma_{pm}^{(\text{W})}(d)$ is monotonically increasing $\gamma
_{pm}^{(\text{W})}(d)\uparrow1$, while The Barrett bound $\gamma_{gm}^{(B)}$ is
monotonically decreasing $\gamma_{gm}^{(\text{B})}(d)$ $\downarrow0$.

For our further construction of $1\times S_{2}$--setting and $S_{1}\times
1$--setting source operators (Definition \ref{source-operator}) for a Werner state
$W_{d,\Phi}$ and the analysis of their tensor positivity (Definition \ref{tensor-positive})
in order to apply the sufficient conditions of Theorem \ref{theorem-1}, we
recall \cite{19,20} the main three mutually orthogonal subspaces of the complex Hilbert space 
$\mathcal{H}_{d}^{\otimes m}$, $d\geq2, m\geq2.$

\emph{The\ fully symmetric subspace of }$\mathcal{H}_{d}^{\otimes m}$ -- the
subspace of vectors in $\mathcal{H}_{d}^{\otimes m}$ which are invariant under
all permutations of tensor factors$.$ The projection $\mathrm{P}_{d,m}^{(+)}$
onto this subspace and its rank are given by\footnote{For
$m=3,$ the explicit expressions for $\mathrm{P}_{d,3}^{(\pm)}$\ are presented,
for example, in \cite{17}.}%
\begin{equation}
\mathrm{P}_{d,m}^{(+)}=\frac{1}{m!}\sum_{\sigma\in\mathfrak{S}m}P_{\sigma
},\text{ \ \ }\mathrm{tr}[\text{ }\mathrm{P}_{d,m}^{(+)}]=\binom{d+m-1}%
{m},\label{7}%
\end{equation}
where $\mathfrak{S}_{m}$ is the symmetric group of permutations on set
$\{1,...,m\}$ and $P_{\sigma}$ is the permutation operator on $\mathcal{H}%
_{d}^{\otimes m}$ defined by
\begin{align}
P_{\sigma}\left(  |x_{1}\rangle\otimes\cdots|x_{m}\rangle\right)    &
:=|x_{\sigma^{-1}(1)}\rangle\otimes\cdots\otimes|x_{\sigma^{-1}(m)}%
\rangle,\label{8}\\
|x_{i}\rangle & \in\mathcal{H}_{d}.\nonumber
\end{align}

\emph{The fully antisymmetric subspace of }$\mathcal{H}_{d}^{\otimes m}$ --
the subspace of vectors in $\mathcal{H}_{d}^{\otimes m}$ changing under a
permutation $\sigma\in\mathfrak{S}_{m}$  of tensor factors according to the
sign $\mathrm{sgn(}\sigma)$ of this permutation. The projection $\mathrm{P}%
_{d,m}^{(-)}$ onto this subspace and its rank 
are given by%
\begin{equation}
\mathrm{P}_{d,m}^{(-)}=\frac{1}{m!}\sum_{\sigma\in\mathfrak{S}m}%
\mathrm{sgn(}\sigma)P_{\sigma},\text{ \ \ }\mathrm{tr}[\mathrm{P}_{d,m}%
^{(-)}]=\binom{d}{m}\text{.}\label{9}%
\end{equation}
The dimension of this subspace is nonzero iff $d\geq m$ and zero otherwise.

\emph{The subspace of mixed symmetry} with projection $\mathrm{P}%
_{d,m}^{(\text{mix})}$. The dimension of this subspace is nonzero for $m\geq3.$ 

The above 
three projections are mutually orthogonal and sum up to identity on
$\mathcal{H}_{d}^{\otimes m}$:%
\begin{equation}
\mathrm{P}_{d,m}^{(+)}\,\,+\,\mathrm{P}_{d,m}^{(-)}\,\,+\,\mathrm{P}%
_{d,m}^{(\text{mix})}=\,\mathbb{I}_{d}^{\otimes m}. \label{10}%
\end{equation}

For an orthogonal projection $\mathrm{P}_{d,m}$ on $\mathcal{H}_{d}^{\otimes
m}$, with the properties specified below, we have the following general statement.

\begin{proposition}\label{lemma-first} Let $A^{(1)}_{d,m}$ and $A^{(2)}_{d,m}$ be Hermitian operators on
$\mathcal{H}_{d}^{\otimes m}, m\geq 3$, with $\mathrm{tr}{[A^{(i)}_{d,m}]}\not=0$, for which the partial traces 
\begin{align} 
\MoveEqLeft[1.5]Q^{(1)}_{1j}:=\mathrm{tr}_{\{2,\ldots,m\}\setminus\{j\}}[A^{(1)}_{d,m}],\; \;\; \; j\in \{2,...,m\},\label{new}\\
\MoveEqLeft[1.5]Q^{(2)}_{km}:=\mathrm{tr}_{\{1,\ldots,m-1\}\setminus\{k\}}[A^{(2)}_{d,m}],\; k\in \{1,...,m-1\},\label{new_}
\end{align}
are independent on $j\in \{2,...,m\}$ and $k\in \{1,...,m-1\}$, respectively, and are $U_{d}\otimes U_{d}$ invariant
\cite{11} for any unitary $U_{d}$ on $\mathcal{H}_{d}$. Then
\begin{align}
\MoveEqLeft[1.5] Q^{(1)}_{1j} =\alpha_{1}^{(+)}\mathrm{P}_{d,2}^{(+)}+\alpha_{1}^{(-)}\mathrm{P}_{d,2}^{(-)},\label{10-1} \\[1ex] 
\MoveEqLeft[1.5] Q^{(2)}_{km}=\alpha_{2}^{(+)}\mathrm{P}_{d,2}^{(+)}+\alpha_{2}^{(-)}\mathrm{P}_{d,2}^{(-)},\label{10-21} \\[1ex] 
\MoveEqLeft[1.5] \alpha_{i}^{(\pm)} =\frac{\mathrm{tr}[A^{(i)}_{d,m}]}{d(d\pm1)}\left( 1\pm\tau(A^{(i)}_{d,m})\right), \label{10-2}
\end{align}
for any $j\in \{2,...,m\}$ and any $k\in \{1,...,m-1\},$ where $\mathrm{P}_{d,2}^{(\pm)}$ are orthogonal projections onto the symmetric and antisymmetric subspaces of $\ \mathcal{H}_{d}^{\otimes2}$ and
\begin{align}
\MoveEqLeft[1.5] \tau(A^{(i)}_{d,m}):=\frac{1}{m-1}\frac{\mathrm{tr}[J^{(i)}_{d,m}
A^{(i)}_{d,m}]}{\mathrm{tr}[A^{(i)}_{d,m}]},\label{10-3}\\
\MoveEqLeft[1.5] J^{(1)}_{d,m}   :=\sum_{j=2}^{m}F_{1j},\text{ \ \ }J^{(2)}_{d,m}:=\sum_{k=1}^{m-1}F_{km}.\text{ \ \ }\nonumber
\end{align}
Here, each $F_{ln}$ is a flip operator on $\mathcal{H}_{d}^{\otimes m},$
swapping the l-th and $n$-th tensor factors.
\end{proposition}

\begin{proof}
Let us prove the statement for  operator $A^{(1)}_{d,m}$. By the setting $Q_{1j_{1}}^{(1)}=Q_{1j_{2}}^{(1)},$ $\forall j_{1},j_{2}\geq2$ and $U_{d}\otimes U_{d})Q_{1j}^{(1)}(U_{d}\otimes U_{d})^{\dagger}=Q_{1j}^{(1)}.$ The latter property implies
\cite{10}
\begin{equation}
Q_{1j}^{(1)}=\alpha_{1}^{(+)}\mathrm{P}_{d,2}^{(+)}+\alpha
_{1}^{(-)}\mathrm{P}_{d,2}^{(-)}:=Q^{(1)},\text{ \ \ }\forall
j\geq2, \label{10-31}%
\end{equation}
i.e. decomposition (\ref{10-1}). Let us calculate $\mathrm{tr}[J^{(1)}_{d,m}A^{(1)}_{d,m}],$ where $J^{(1)}_{d,m}=\sum_{j=2}^{m}F_{1j}$ and $F_{1j}$ is the
flip operator on $\mathcal{H}_{d}^{\otimes m}$ swapping the first and
$j$-th tensor factors. By (\ref{10-31}) and the structure of operator $J^{(1)}_{d,m}$ we have
\begin{align}
\mathrm{tr}[J^{(1)}_{d,m}A^{(1)}_{d,m}]  &  =\sum_{j=2}^{m}\mathrm{tr}%
[F_{1j}A^{(1)}_{d,m}]=\sum_{j=2}^{m}\mathrm{tr}[V_{1j}Q^{(1)}_{1j}]
\label{10-32}\\
&  =(m-1)\mathrm{tr}[V_{12}Q^{(1)}],\nonumber
\end{align}
where $V_{12}$ is the flip operator on $\mathcal{H}_{d}\otimes\mathcal{H}%
_{d}.$ Calculating further $\mathrm{tr}[V_{12}Q^{(1)}_{1j}]$ and 
$\mathrm{tr}[A^{(1)}_{d,m}]$, we come to the relations
\begin{align}
\MoveEqLeft[1.5]\mathrm{tr}[A^{(1)}_{d,m}]=\alpha_{1}^{(+)}%
\frac{d(d+1)}{2}+\alpha_{1}^{(-)}\frac{d(d-1)}%
{2},\label{10-4}\\
\MoveEqLeft[1.5]\mathrm{tr}[J^{(1)}_{d,m}A^{(1)}_{d,m}]=(m-1)\left(  \alpha_{1}^{(+)}\frac{d(d+1)}{2}-\alpha_{1}^{(-)}\frac
{d(d-1)}{2}\right)\nonumber
\end{align}
which imply%
\begin{equation}
\alpha_{1}^{(\pm)}=\frac{1}{d(d\pm1)}\left(\mathrm{tr}%
[A^{(1)}_{d,m}]\pm\frac{1}{m-1}\mathrm{tr}[J^{(1)}_{d,m}A^{(1)}_{d,m}]\right),
\label{10-6}%
\end{equation}
This proves relations (\ref{10-2}), (\ref{10-3}) for operator $A^{(1)}_{d,m}$. The proof for operator $A^{(2)}_{d,m}$ is quite similar.  
\end{proof}

For operators $J^{(i)}_{d,m},\;  i=1,2,$ in (\ref{10-3}), the eigenvalues $\lambda$ coincide and satisfy  the relation\footnote{This  follows immediately from the standard content-eigenvalue theorem and the Schur–Weyl duality, see, e.g., in \cite{21,22} and references therein.} 
\begin{equation}
-\min \{d-1,m-1\}\leq \lambda\leq m-1.\label{10-8}
\end{equation}
The spectral projections $\mathrm{P}^{(i)}_{\lambda}$, 
corresponding to the same eigenvalue $\lambda$, are unitary equivalent, satisfy the setting of Proposition \ref{lemma-first} and have the same parameter   $\tau(\mathrm{P}^{(i)}_{\lambda})%
=\frac{\lambda}{m-1}\leq1.$

This, in particular, refers to: (i) the spectral projections, corresponding to the maximal eigenvalue $\lambda_{\max}(d,m)=(m-1)$ of operators $J^{(i)}_{d,m}, i=1,2,$ which coincide and are given by the orthogonal projection  (\ref{7}) onto the fully symmetric subspace of $\mathcal{H}_{d}^{\otimes m}$; (ii) the spectral projections $\mathrm{P}_{d,m}^{(\min)}(i)$, corresponding to the minimal eigenvalue $\lambda_{\min}(d,m)=-\min\{d-1,m-1\}$ of operators $J^{(i)}_{d,m}$. For  $d\geq m$, the latter spectral projections coincide and are given by the  orthogonal projection $\mathrm{P}_{d,m}^{(-)}$ onto the fully antisymmetric subspace of 
$\mathcal{H}_{d}^{\otimes m}$, presented in Eq. (\ref{9}). For the orthogonal projections (\ref{7}) and (\ref{9}), parameters $\tau(\mathrm{P}_{d,m}^{(\pm)})=\pm1,$ respectively.

Relations (\ref{10-3}), (\ref{10-8}) and the bound\footnote{Here, $||\cdot||_{0}$ means
the operator norm.}  $\left\vert\mathrm{tr}[AB]\right\vert \leq||A||_{0}\mathrm{tr}[B]$, valid for any operator $A$ and any positive operator $B$, imply that, for any positive\footnote{On a finite-dimensional complex Hilbert space every positive operator is Hermitian.} operator $Z_{d,m}$ on $\mathcal{H}_{d}^{\otimes m}$ ,  satisfying either of the two settings $i=1,2$ in Proposition \ref{lemma-first}, the parameter
\begin{equation}
-\frac{\min \{d-1,m-1\}}{m-1}\leq \tau(Z_{d,m})\leq1,\label{10-9}
\end{equation}
where the left-hand side is attained at the spectral projections $\mathrm{P}_{d,m}^{(\min)}(i)$.

By Definition \ref{source-operator} and the properties \cite{10} of the Werner states (\ref{4}), the source operators $T^{W_{d,\Phi}}_{1\times S_{2}
}$ and $T^{W_{d,m}}_{S_{1} \times 1}$ of a Werner state $W_{d,\Phi}$, discussed in Theorem  \ref{theorem-1},  satisfy the setting of Proposition \ref{lemma-first}, have the unit trace and parameters 
\begin{equation}
\tau(T^{W_{d,\Phi}}_{1\times S_{2}})=\tau(T^{W_{d,\Phi}}_{S_{1}\times 1})=\Phi\leq1.
\end{equation}
If either of these source operator $T$ is positive, then its parameter $\tau(T)$ satisfies relation (\ref{10-9}).

For our further proof of Theorem \ref{theorem-2}, let us introduce the following Hermitian operators on $\mathcal{H}_{d}^{\otimes{m}}, m\geq3$:
\noindent\begin{align}
\MoveEqLeft[3.0]X_{d,m}^{(i)}(r):=(J_{d,m}^{(i)})^{2}+(1-2r)J_{d,m}^{(i)}+(r^{2}-n)
\mathbb{I}_{\mathcal{H}_{d}^{\otimes {m}}},\label{_2} 
\end{align}
where  $r\in\mathbb{R}$, $n=m-1$  and operators  $J^{(i)}_{d,m},  i=1,2$, are defined in (\ref{10-3}). These Hermitian operators have the same eigenvalues, satisfy the setting of Proposition \ref{lemma-first} and, as shown in Appendix, are tensor positive $X^{(i)}_{d,m}(r)\overset{\otimes }{\geq }0$  for all $r\in\mathbb{R}$. 

By specifying the intervals of $r$, where  operators $X^{(i)}_{d,m}(r), i=1,2,$ have negative eigenvalues, we find that, for $2\leq d\leq n,$ where $n=m-1\geq 2,$ the tensor positive operators $X^{(i)}_{d,m}(r)\overset{\otimes }{\geq }0$ are
not positive semi-definite, i.e. $X^{(i)}_{d,m}(r)\ngeqslant 0,$ iff 
\begin{equation}
-(d-1)-\sqrt{n+d-1} < r <n.\label{0}
\end{equation}
For $d>n,$ operators $X^{(i)}_{d,m}(r)\overset{\otimes }{\geq }0$, \ $X^{(i)}_{d,m}\ngeqslant 0,$  iff 
\begin{equation}
-n-\sqrt{2n}<r<n.\label{00}
\end{equation}

For operators $X^{(i)}_{d,m}(r)\overset{\otimes }{\geq }0$, parameters (\ref{10-3}) coincide and are given by
\begin{align}
\tau(X^{(i)}_{d,m}(r))=\frac{1}{d}-\frac{(d^{2}-1)(2s-d+2)}{d(s^{2}+n(d-1))}, \;\;i=1,2,\label{_3}
\end{align}
where $s:=dr-n,  n=m-1$. The minimum of this expression over variable $r$ is attained at $r_{0}=(s_{0}+n)/d,$ where
\begin{equation}
s_{0}=\frac{d-2+\sqrt{(d-2)^{2}+4n(d-1)}}{2}\label{_4}
\end{equation}
is the root of the quadratic equation $s^{2}-(d-2)s-n(d-1)=0.$ This minimum is equal to 
\begin{eqnarray}\label{000}
\gamma(d,n) &:=&\tau (X_{d,m}^{(i)}(r_{0}))\\
&=&\frac{1}{d}\left( 1-\frac{2(d^{2}-1)}{\sqrt{(d-2)^{2}+4n(d-1)}+(d-2)}%
\right),\notag
\end{eqnarray}
where $\gamma(d,n)< \frac{1}{d}, n=m-1$. It is also easy to show that, for $m\geq3$, parameter $\gamma(d,n)>-1$.

The normalized versions of operators $X^{(i)}_{d,m}$ constitute the source operators $T^{W_{d,\Phi}}_{1\times S_{2}
}$ and $T^{W_{d,m}}_{S_{1} \times 1}$ of a Werner state $W_{d,\Phi}$. Since $0<r_{0}<n$ belongs to the intervals, specified in (\ref{0}) and (\ref{00}), these tensor positive source operators are not positive semi-definite.

In view of Definition \ref{Bell-locality}, Theorem \ref{theorem-1}, Proposition \ref{lemma-first}, and relations (\ref{10-8})--(\ref{000}), we proceed to prove the following general statement.
\begin{theorem}
\label{theorem-2} Let $W_{d,\Phi}$ be a nonseparable Werner state on
$\mathcal{H}_{d}\otimes\mathcal{H}_{d},$ $d\geq2,$ and $S_{1},S_{2}\geq1$ be
arbitrary positive integers with\footnote{Any correlation scenario with
settings $1\times S_{2}$ or $S_{1}\times1,$ that is, with $S_{\min}=1,$ admits
the LHV description, see Proposition 2 in \cite{3}.} $S_{\min}:=\min
\{S_{1},S_{2}\}\geq2$ and $S_{\max}:=\max\{S_{1},S_{2}\}$. The following
statements hold.
\smallskip\newline(\textbf{a}) For any $2\leq d\leq S_{\min},$ a
nonseparable Werner state $W_{d,\Phi}$ is $S_{1}\times S_{2}$--setting Bell local if
\begin{align}
\MoveEqLeft[1.5] 
\Phi \in \left[ \gamma _{0}, 0 \right), \; \gamma _{0}=\min \left\{-\frac{d-1}{S_{\min }}, \;\gamma(d,S_{\min })\right\} <0,\label{10"}
\end{align}
moreover,  is $S_{1}\times
L_{2}^{\prime}$--setting Bell local for all\footnote{See footnote 3.} $L_{2}^{\prime}>S_{2}$ if
$S_{\min}=S_{1}$ and $L_{1}^{\prime}\times S_{2}$--setting Bell local for all
$L_{1}^{\prime}>S_{1},$ if $S_{\min}=S_{2}$. Here, $-1<\gamma(d,S_{\min })<\frac{1}{d}$ is given by (\ref{000}) with $n=S_{\min}$.
\medskip\newline(\textbf{b}) For $S_{\min}<d\leq S_{\max}$,
every nonseparable Werner state $W_{d,\Phi},$ $\Phi\in\lbrack-1,0),$ is
$S_{1}\times S_{2}$--setting Bell local, moreover, is $S_{1}\times
L_{2}^{\prime}$--setting Bell local for all $L_{2}^{\prime}>S_{2}$ if
$S_{\min}=S_{1}$ and $L_{1}^{\prime}\times S_{2}$--setting Bell local for all
$L_{1}^{\prime}>S_{1},$ if $S_{\min}=S_{2}.$
\medskip\newline(\textbf{c}) If
$d>S_{\max},$ every nonseparable $W_{d,\Phi},$ $\Phi\in\lbrack-1,0),$ is
$S_{1}\times S_{2}$--setting Bell local, moreover, is $S_{1}\times
L_{2}^{\prime}$--setting Bell local for all $L_{2}^{\prime}>S_{2}$ and
$L_{1}^{\prime}\times S_{2}$--setting Bell local for all $L_{1}^{\prime}%
>S_{1}$.
\end{theorem}
\begin{proof}
(1) For the construction of a $1\times S_{2}$-setting source operator for a
nonseparable Werner state $W_{d,\Phi}$, consider the following Hermitian
operator on space $\mathcal{H}_{d}\otimes\mathcal{H}_{d}^{\otimes S_{2}}:$%
\begin{equation}
\eta\mathrm{P}_{d,S_{2}+1}^{(+)}+\xi\ A^{(1)}_{d,S_{2}+1},\text{
\ \ }\eta,\xi\in\mathbb{R},\label{11_}%
\end{equation}
where $\mathrm{P}_{d,S_{2}+1}^{(+)}$ is projection (\ref{7}) and $A^{(1)}_{d,S_{2}+1}$ is a Hermitian operator on $\mathcal{H}_{d}^{\otimes(S_{2}+1)}$ with $\mathrm{tr}{[A^{(1)}_{d,m}]}>0$, which satisfies the setting of Proposition \ref{lemma-first} and has parameter $-1\leq\tau(A^{(i)}_{d,m})<1$. By (\ref{7}) and (\ref{10-1})
we have the relations%
\begin{align}
\mathrm{tr}_{\{2,\ldots,S_{2}+1\}\setminus\{j\}}
[\mathrm{P}_{d,S_{2}+1}^{(+)}]
&=
\binom{d+S_{2}}{S_{2}+1}\,
\frac{\mathrm{P}_{d,2}^{(+)}}{r_{d}^{(+)}},
\label{12}\\
\mathrm{tr}_{\{2,\ldots,S_{2}+1\}\setminus\{j\}}
[A^{(1)}_{d,S_{2}+1}]
&=
\alpha_{1}^{(+)}\mathrm{P}_{d,2}^{(+)}
+\alpha_{1}^{(-)}\mathrm{P}_{d,2}^{(-)},
\nonumber
\end{align}
which are valid for any $j\geq2.$ Therefore, by Definition \ref{source-operator} the
Hermitian operator 
\begin{equation}
T_{1\times S_{2}}^{(W_{d,\Phi})}(A^{(1)}_{d,S_{2}+1}):=\text{ }%
\gamma\,\mathrm{P}_{d,S_{2}+1}^{(+)}\,+\,\xi\ A^{(1)}_{d,S_{2}+1}\label{15}%
\end{equation}
constitutes a $1\times S_{2}$ source operator for a Werner state $W_{d,\Phi},$ $d\geq2,$ iff
\begin{align}
&\eta\,\binom{d+S_{2}}{S_{2}+1}\,
\frac{2\mathrm{P}_{d,2}^{(+)}}{d(d+1)}
+\xi\left(
\alpha_{1}^{(+)}\mathrm{P}_{d,2}^{(+)}
+\alpha_{1}^{(-)}\mathrm{P}_{d,2}^{(-)}\right)\label{16}\\
&=\text{}W_{d,\Phi}.\nonumber
\end{align}
Substituting (\ref{4}) into (\ref{16}) and taking into account (\ref{10-2}), we
derive:
\begin{align}
\MoveEqLeft[1.5]\xi= \frac{1}{\mathrm{tr}[A^{(i)}_{d,m}]} \frac{\makebox[\widthof{\(1-\tau(A^{(1)}_{d,m})\)}][c]{\(1-\Phi\)}}{\;\;\;1-\tau(A^{(1)}_{d,m})},\;
 \eta= \frac{\Phi-\tau(A^{(1)}_{d,m})}{1-\tau(A^{(1)}_{d,m})} \;\binom{d+S_{2}}{S_{2}+1}^{-1}\label{hh}
\end{align}
These relations and the above restrictions on an operator $A^{(1)}_{d,m}$, taken in (\ref{11_}), imply that $\xi\geq0$ and 
$\eta \geq 0\Leftrightarrow \Phi \in \lbrack \tau(A^{(1)}_{d,m}), 1]$, 
where $\tau (A_{d,m}^{(1)})=\frac{1}{S_{2}}\frac{\mathrm{tr}%
[J_{d,m}^{(1)}A_{d,m}^{(1)}]}{\mathrm{tr}[A_{d,m}^{(1)}]}.$
Thus, the operator (\ref{15}) constitutes a $1\times S_{2}$-setting source operator  for the Werner state $W_{d,\Phi}$ with parameter $\Phi\in [-1,1]$ if its coefficients $\xi$ and $\eta$ satisfy relations (\ref{hh}). Furthermore, if, in (\ref{15}), we take $A_{d,m}^{(1)}\overset{\otimes }{\geq }0$, then, for each Werner state $W_{d,\Phi}$ with parameter $\Phi \in \lbrack \tau(A^{(1)}_{d,m}),1],$  the corresponding $1\times S_{2}$-setting source operator (\ref{15}) is tensor positive.
\smallskip
\newline(1.a) Let $2\leq d\leq S_{2}.$ Take $A^{(1)}_{d,m}$ in (\ref{15}) as the orthogonal projection $\mathrm{P}_{d,S_{2}+1}^{(\min)}$, corresponding to the minimum eigenvalue $\lambda_{\min}$ of operator $J^{(1)}_{d,S_{2}+1}$, see in Eq.(\ref{10-8}). In this case,  $\tau(\mathrm{P}_{d,S_{2}+1}^{(\min)})=-\frac{d-1}{S_{2}}$ and, for each Werner state with $\Phi\in\lbrack
-\frac{d-1}{S_{2}},1],$ the $1\times S_{2}$-setting source operator (\ref{15}) with $A^{(1)}_{d,m}=\mathrm{P}_{d,S_{2}+1}^{(\min)}$ is positive semi-definite. 
If, however, we take operator $A^{(1)}_{d,m}$ as the tensor positive operator $X_{d,m}^{(1)}(r_{0})\overset{\otimes }{\geq }0$ in (\ref{_2}), which is not positive semi-definite (see  relations (\ref{_2})--(\ref{000})), then, for each Werner state with $\Phi\in\lbrack
\gamma(d,s_{2}),1]$, where $\gamma(d,s_{2})$ is given by Eq. (\ref{000}), the corresponding $1\times S_{2}$-setting source operator (\ref{15}) is tensor-positive but not positive semi-definite.

\noindent In the result, if $d\leq S_{2}$, then for every nonseparable Werner state $W_{d,\Phi }$ with 
\begin{equation}
\Phi \in\left[\min\left\{-\frac{d-1}{S_{2}},\;\gamma
(d,S_{2})\right\} ,0\right),
\end{equation}
there is a $1\times S_{2}$-setting tensor positive source operator, which is not positive semi-definite if $\gamma(d,S_{2})<-\frac{d-1}{S_{2}}.$ 
\smallskip
\newline(1.b) Let $d>S_{2}$.
In this case, the dimension of the antisymmetric
subspace of $\mathcal{H}_{d}^{\otimes(S_{2}+1)}$ is nonzero, projection
$\mathrm{P}^{(\min)}_{d,S_{1}+1} =\mathrm{P}_{d,S_{2}+1}
^{(-)}$, parameter $\tau_{\mathrm{P}_{d,S_{2}+1}^{(-)}}=-1$ and from (\ref{15}), (\ref{hh}) it follows that 
\begin{align}
\widetilde{T}_{1\times S2}^{(W_{d,\Phi})}(\mathrm{P}_{d,S_{2}+1}^{(-)}):=\widetilde{\eta}\,\mathrm{P}_{d,S_{2}+1}^{(+)}\text{ }+\text{
}\widetilde{\xi}\,\mathrm{P}_{d,S_{2}+1}^{(-)},\label{23}
\end{align}
is a $1\times S_{2}$-setting positive semi-definite source operator for each nonseparable Werner state $W_{d,\Phi}, \Phi\in\lbrack-1,1]$ if
\begin{align}
\MoveEqLeft[1.5]\widetilde{\xi}=\frac{1-\Phi}{2}\binom{d}{S_{2}+1}^{-1}\geq0, \;
\widetilde{\eta}=\frac{1+\Phi}{2}\binom{d+S_{2}}{S_{2}+1}^{-1}
\geq0.\nonumber
\end{align}
\smallskip
\newline(2) Quite similarly  as in item (1), choosing $A^{(2)}_{d,S_{1}+1}$ as the spectral projection $\mathrm{P}_{d,S_{1}+1}^{(\min)}$ of operator $J^{(2)}_{d,S_{1}+1}$ if $d\leq S_{1}$ (like in item (1a)) and as the projection $\mathrm{P}_{d,S_{1}+1}
^{(-)}$ (like in item (1,b)), we prove that, if $d\leq S_{1}$, then, for every nonseparable Werner state $W_{d,\Phi }$ with 
\begin{equation}
\Phi \in\left[\min\left\{-\frac{d-1}{S_{1}},\;\gamma
(d,S_{1})\right\} ,0\right)\label{m}
\end{equation}
there is a $S_{1}\times 1$-setting tensor positive source operator, which is not positive semi-definite if $\gamma(d,S_{1})<-\frac{d-1}{S_{1}}$.
\noindent Similarly as in item (1.b), for $d > S_{1}$, the positive semi-definite  $S_{1}\times1$-setting source operator exists for any nonseparable Werner state.
\newline\emph{In view of Theorem \ref{theorem-1}, combining the  results of items (1) and (2) proves the statements of Theorem
\ref{theorem-2}.}
\end{proof}

We stress that, for $d\leq S_{\min}$, our new general sufficient condition (\ref{10"}) improves upon the related result in Theorem 6 of \cite{23} whenever
\begin{equation}
\gamma(d,S_{\min})<-\frac{d-1}{S_{\min}}.
\end{equation}
We have: (i) if $d=2$, then the relation
\begin{equation}
\gamma(2,S_{\min}) =\frac{1}{2}\Big(1-\frac{3}{\sqrt{S_{\min}}}\big)< -\frac{1}{S_{\min}}\label{k}
\end{equation}
holds for $S_{\min}=2,3$; (ii) 
if $d\geq3$, then the relation
$\gamma(d,S_{\min})<-\frac{d-1}{S_{\min}}$ is true for all 
\begin{equation}
2d\leq S_{\min}<d^{2}(d-2)+1.\label{kk}
\end{equation}

Item (i) implies that, specified for $d=2$ and $S_{\min}=2,3$, our general bound (\ref{10"}) improves the results of Terhal et. al. [16] for a two-qubit Werner state and these settings, moreover, \emph{for $d=S_{\min}=2$, our bound $\Phi\geq \frac{1}{2}\left(1-\frac{3}{\sqrt{2}}\right)$ 
coincides exactly with the bound derived for this case via the CHSH inequality.}

\section{Conclusion}
In the present article, we have derived analytically a new general condition (\ref{10"})
sufficient for a nonseparable Werner state of a dimension $d\leq\min
\{S_{1},S_{2}\}$ to satisfy all Bell inequalities under any $S_{1}\times
S_{2}$-setting correlation scenario with outcomes of any spectral type, discrete or continuous -- that is, to  be $S_{1}\times S_{2}$-setting Bell local, for short. 

For a variety of measurement  settings $S_{1},S_{2}\geq1$, this
new general sufficient condition on $S_{1}\times S_{2}$-setting Bell locality  is beyond the generally used sufficient locality conditions (\ref{5}),
(\ref{6}) by Werner \cite{10} for scenarios with projective parties'
measurements and by Barrett \cite{11} for scenarios with generalized parties' measurements. 
For example, for $d=2$, the Barrett's locality bound $\gamma_{gm}%
^{(B)}(2)=\frac{1}{8}$ in (\ref{5}) is less than the value $|\gamma_{0}|$ in 
bound (\ref{10"}) under all $S_{1}\times S_{2}%
$--setting correlation scenarios with $\min\{S_{1},S_{2}\}\leq7$. For $d=3,$
the Barrett's bound $\gamma_{gm}^{(B)}(3)=\frac{5}{81}$ is less than the value
$|\gamma_{0}|$ in  (\ref{10"}) under all scenarios with $\min\{S_{1},S_{2}%
\}\leq32.$ 

When specialized to $d = 2$ and $\min\{S_{1},S_{2}\} = 2, 3$, our new general locality bound (\ref{10"}) improves due to (\ref{k}) upon the corresponding results by Terhal et. al. in \cite{15}, moreover, for $d=S_{\min}=2$, \emph{ it reduces exactly to the bound derived for this case via the CHSH inequality}. The analytical
expression (\ref{23}) for $1\times S_{2}$-setting 
positive source operator of a Werner state with $d>S_{\min}$ proves also explicitly in operator terms the result by Terhal et. al \cite{15} for this case.

When specialized to $3\leq\ d\leq S_{\min}$ and $S_{\min}$ in the interval (\ref{kk}), the new locality bound (\ref{10"}) improves upon the related "sharability" result in \cite{23}.

The new results presented in this article are important for both the theory of Bell nonlocality and quantum applications based on this fundamental concept.

\section*{acknowledgments}
The study in Sections III, V  of this work was supported by the Russian Science Foundation under the Grant \textnumero 24-11-00145,
https://rscf.ru/en/project/24-11-00145/ and performed at the Steklov Mathematical Institute of Russian Academy of Sciences. The study in Sections I, II, IV was performed at the HSE University. The valuable comments of the anonymous Referees are highly appreciated.
\section{appendix}
In what follows we prove that the Hermitian operators (\ref{_2}) are tensor positive for all $r\in \mathbb{R}$ and not positive semi-definite under relations (\ref{0}), (\ref{00}). By Definition 
\ref{tensor-positive} operator $X^{(i)}_{d,m}(r)\overset{\otimes }{\geq }0$ iff $
\langle\psi|X_{d,,m}^{(i)}|\psi\rangle \geq 0$,
for all unit product vectors $|\psi\rangle:=|\psi _{1}\otimes \cdots \otimes \psi
_{m}\rangle ,$ $|\psi _{k}\rangle \in \mathcal{H}_{d}.$ 

Consider the proof of this relation for $
X_{d,,m}^{(1)}(r).$ 
Taking into account that $
J_{d,m}^{(1)}=
\sum_{j=2}^{m}F_{1j}, \;F_{1j}^{2}=\mathbb{I}_{d}^{\otimes m},\;
(J_{d,m}^{(1)}) ^{2}=n\mathbb{I}_{d}^{\otimes m}+\sum_{j\neq j_{1}\geq 2}F_{1j}F_{1j_{1}}$, we have:
\begin{equation}
\langle\psi|J_{d,m}^{(1)}|\psi\rangle = \langle \psi_{1}|G_{\psi }^{(1)}|\psi _{1}\rangle, \quad G_{\psi}^{(1)}=\sum_{k=2}^{m}|\psi _{k}\rangle \langle \psi _{k}|,\label{A3} 
\end{equation}
\begin{flalign}
\langle \psi |\left( J_{d,m}^{(1)}\right) ^{2}|\psi \rangle = n + \langle \psi_{1}|\left( G_{\psi }^{(1)}\right) ^{2}|\psi_{1}\rangle - \langle\psi_{1}|G_{\psi}^{(1)}|\psi_{1}\rangle. \notag
\end{flalign} Substituting (\ref{A3}) relations into (\ref{_2}), we derive $
\langle \psi |X_{d,,m}^{(1)}|\psi \rangle=\langle \psi _{1}|\left( G_{\psi }^{(1)}-r\mathbb{I}_{d}\right) ^{2}|\psi
_{1}\rangle \geq 0$. This proves $X_{d,,m}^{1)}(r)\overset{\otimes }{\geq }0$ for all $r\in \mathbb{R}$. The proof of $
X_{d,,m}^{(2)}(r)\overset{\otimes }{\geq }0$ is quite similar.

To prove relations (\ref{0}), 
(\ref{00}) we note that operator $X_{d,,m}^{(i)}(r)\ngeq 0$ iff at least one of its eigenvalues $
\xi (\lambda ,r,n)=(\lambda -r)^{2}-(n-\lambda )$ is negative. Here, $\lambda $ is eigenvalue of $J_{d,m}^{(i)}$. Negativity of $\xi (\lambda ,r,n)$ implies $
\lambda -\sqrt{n-\lambda }<r<\lambda +\sqrt{n-\lambda}\
$. From this relation and Eq. (\ref{10-8}) it follows: (i) for $d\leq n$,  the tensor positive operator $X_{d,,m}^{(i)}(r)$ is not positive semi-definite iff $r\in\left(-(d-1)-\sqrt{n+d-1}, n\right)$; (ii) for $d>n$, the tensor positive operator $X_{d,,m}^{(i)}(r)$ is not positive semi-definite iff $r\in\left(-n-\sqrt{2n},n\right).$  
This proves relations (\ref{0}) and (\ref{00}).

\end{document}